\documentclass[sts]{imsart}

\RequirePackage{amsthm,amsmath,amsfonts,amssymb}
\usepackage[authoryear,round]{natbib}
\usepackage[colorlinks=true,linkcolor=blue,citecolor=blue,urlcolor=blue]{hyperref}
\RequirePackage{graphicx}
\usepackage{tikz, subcaption}
\usetikzlibrary{graphs, positioning, calc}
\usepackage{tabularx} 
\usepackage{enumitem} 

\usepackage{xcolor}
\usepackage{xcolor}

\usepackage[normalem]{ulem}
\newcommand{\stkout}[1]{\ifmmode\text{\sout{\ensuremath{#1}}}\else\sout{#1}\fi}

\tikzset{
  >=latex, 
  every path/.style={line width=1.2pt}, 
  every arrow/.append style={scale=1.5}, 
  dotted/.append style={line width=1.2pt} 
}

\newlist{titemize}{itemize}{1}
\setlist[titemize]{nosep,left=0pt,label=--}

\startlocaldefs
\theoremstyle{plain}

\newtheorem{theorem}{Theorem}[section]

\theoremstyle{definition}
\newtheorem{definition}[theorem]{Definition}

\newtheorem{assumption}{Assumption}
\endlocaldefs

\begin{document}

\begin{frontmatter}
\title{Comparing Two Proxy Methods for Causal Identification}
\begin{aug}
\author[A]{\fnms{Helen}~\snm{Guo}\ead[label=e1]{hguo51@jh.edu}}, \author[A]{\fnms{Elizabeth L.}~\snm{Ogburn}\ead[label=e2]{eogburn@jhu.edu}},
\and
\author[B]{\fnms{Ilya}~\snm{Shpitser}\ead[label=e3]{ilyas@cs.jhu.edu}}

\address[A]{Department of Biostatistics,
Johns Hopkins Bloomberg School of Public Health \printead[presep={,\ }]{e1,e2}}

\address[B]{Department of Computer Science,
Johns Hopkins Whiting School of Engineering \printead[presep={,\ }]{e3}}
\end{aug}

\begin{abstract}
Identifying causal effects in the presence of unmeasured variables is a fundamental challenge in causal inference, for which proxy variable methods have emerged as a powerful solution. We contrast two major approaches in this framework: (1) bridge equation methods, which leverage solutions to integral equations to recover causal targets, and (2) array decomposition methods, which recover latent factors used to identify counterfactual quantities via eigendecomposition tasks. We compare the model restrictions underlying these two approaches and provide insight into implications of the underlying assumptions, clarifying the scope of applicability for each method.
\end{abstract}


\end{frontmatter}

\section{Introduction}\label{sec:1}

Identifying causal effects in the presence of unmeasured variables is a fundamental challenge in causal inference. Interest in proxy variable methods, i.e. methods that leverage variables associated with latent variable(s), has grown rapidly in both applied and methodological domains. Within this landscape, two foundational lines of work underpin proxy-based approaches to nonparametric identification\footnote{In this paper, ``nonparametric identification'' refers to identification under conditional independence constraints and additional functional restrictions, provided these are sufficiently weak in the sense that under relative support size/dimensionality assumptions on the variables, distribution violating them can be approximated arbitrarily well by distributions that satisfy them \citep{semiparametric2020cui, Canay_testability2013}.} of causal effects: (1) the bridge equation approach first introduced by \citet{miao2018identifying}; and (2) the array decomposition approach, introduced to the causal literature by \citet{kuroki14measurement} and \citet{allman2015parameter}, building on ideas for full law identification in latent variable models which trace back to \citet{kruskal1977three}. 

The bridge equation approach was developed for settings in which proxies are available for an unmeasured confounder, allowing the counterfactual distribution to be recovered through re-expression of the standard confounding adjustment formula, without identification of the full law \citep{miao2018identifying}. The method leverages solutions to bridge equations—a generalization of solutions to linear systems—to recover the causal target. This method was later adapted to settings where proxies are available for an unmeasured mediator \citep{ghassami24causal}, recovering causal effects through re-expression of Judea Pearl's front-door adjustment formula \citep{pearl95causal}.

The array decomposition approach imposes invertibility conditions similar to those that ensure solutions to bridge equations, but rather than solving directly for a causal target, it recovers the full joint law of latent and observed variables, from which causal effects may be derived. Originating from the work of \citet{kruskal1977three}, this framework facilitates causal effect identification in models with key hidden variables beyond confounders, such as in settings with unmeasured exposures \citep{zhou2024causalinferencehiddentreatment}. 

Both approaches identify causal effects in the presence of unmeasured confounding under distinct, non-nested conditional independence structures \citep{deaner2023controllinglatentconfoundingtriple}. This naturally raises the question of how the frameworks compare under shared conditional independences, in terms of additional assumptions imposed.

This paper aims to synthesize and explicate the two frameworks, placing them on common footing under shared conditional independences. 
Since the assumptions underlying each framework are technically subtle and often difficult to interpret, clarifying their precise relationship is essential for both methodological development and applied use. Even among leading contributors to the field, there remains active discussion about how the two approaches differ, underscoring the need for careful exposition.

The assumptions underlying much of the existing literature are most transparent in fully discrete settings. In this regime, conditional distributions can be represented as matrices, and identification reduces to invertibility or rank conditions on these finite-dimensional objects. This algebraic perspective provides both technical clarity and intuition. Both frameworks, however, extend to models with continuous variables.

Accordingly, we frequently reference the discrete formulation throughout the paper to build intuition, and then discuss the continuous analogues in detail. Section~\ref{sec:1} provides a brief overview of relevant background. In Section~\ref{sec:2}, we use a matrix representation to elucidate the assumptions underlying the bridge equation approach and to motivate their continuous analogues. Section~\ref{sec:3} begins with the discrete formulation of the array decomposition approach, and then transitions to full law identification in continuous settings as developed by \citet{hu2008instrumental} and extended to causal targets by \citet{deaner2023controllinglatentconfoundingtriple}. Section~\ref{sec:4} compares the two approaches under shared conditional independences in the presence of hidden confounding. Section~\ref{sec:5} briefly discusses estimation strategies, and Section~\ref{sec:6} concludes with directions for future work.

\subsection{Causal Graphical Models}\label{subsec:1a}

Bayesian networks, represented as directed acyclic graphs (DAGs), offer a structured framework for statistical modeling, where nodes denote random variables and edges encode conditional independence assumptions \citep{pearl88probabilistic}. DAGs may be interpreted causally \citep{pearl09causality, spirtes01causation}, encoding counterfactual relationships which enable formal reasoning about causal effects and their identifiability from observed data \citep{richardson13swig}. For example, Figure~\ref{fig:1} depicts a simple confounding structure in which a common cause \(C\) influences both treatment \(A\) and outcome \(Y\). This structure entails the counterfactual independence statement \(Y(a) \perp\!\!\!\perp A \mid C\), known as conditional ignorability, where $Y(a)$ denotes a \emph{potential outcome} random variable, interpreted to mean ``$Y$ had $A$, possibly contrary to fact, been set to value $a$'' \citep{pearl09causality}. 

The methods we discuss rely on conditional independences and counterfactual assumptions reflected in causal graphical models referenced throughout the paper. For an introduction to causal DAGs sufficient to understand this paper, see \citet{elwert2013graphical}.

\subsection{Adjustment Formula}\label{subsec:1b}

Before introducing proximal methods for identification of causal targets, we first review the standard covariate adjustment formula (Equation~\ref{eq:adjustment}). 

Consider the causal graphical model in Figure~\ref{fig:1}, where variables \(C\) creating treatment-outcome (\(A\) and \(Y\), respectively) confounding are observed. The counterfactual distribution \(f_{Y(a)}(y)\) is identified by the standard confounding adjustment formula,

\begin{equation}\label{eq:adjustment}
f_{Y(a)}(y) 
= \int f_{Y \mid A,C}(y \mid a, c) \, f_C(c) \, dc,
\end{equation}
where we use $f$ to denote both full and observed data distributions throughout, provided the following conditions hold for each \(a \in \mathcal{A}\):  

\begin{assumption}[Conditional Ignorability] \label{assump:ignorability}
\[
Y(a) \perp\!\!\!\perp A \mid C.
\]
\end{assumption}

\begin{assumption}[Consistency]
\[Y = Y(a) \text{ when } A = a.\]
\end{assumption}

\begin{assumption}[Positivity]
\[f_{A \mid C}(a \mid c) > 0 \text{ for all } c \in\mathcal{C}.\]
\end{assumption}

Consequently, the average causal effect in the case of binary \(A\) is identified by Equation~\ref{eq:adjustment2}. 

{\small
\begin{equation}\label{eq:adjustment2}
\begin{aligned}
&\quad\quad \mathbb{E}[Y(1)] - \mathbb{E}[Y(0)] =\\
& \int \Big(
      \mathbb{E}[Y \mid A=1, C=c]
      -\mathbb{E}[Y \mid A=0, C=c]
      \Big) f_C(c) \, dc.
\end{aligned}
\end{equation}
}

In general, Assumption~\ref{assump:ignorability} holds if, by the rules of d-separation, conditioning on \(C\) blocks all backdoor paths (which start with an arrowhead into $A$) from \(A\) to \(Y\) \citep{pearl09causality, vanderweele13on, richardson13swig}. 

\begin{figure}[tb]
    \centering
    \begin{tikzpicture}[>=latex, node distance=2cm, anchor=center]
        \node (A) [draw, circle] {A};
        \node (Y) [draw, circle, right of=A] {Y};
        \node (C) [draw, circle, above of=A, xshift=1cm] {C};

        \draw[->] (A) -- (Y);
        \draw[->] (C) -- (Y);
        \draw[->] (C) -- (A);
    \end{tikzpicture}
    \caption{\centering 
Observed confounding; \\
Supports $\mathcal{A}, \mathcal{Y}, \mathcal{C}$}
    \label{fig:1}
\end{figure}
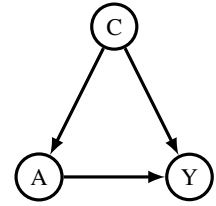

We provide an analogous review of the front-door adjustment formula from  \citet{pearl95causal} in Appendix~\ref{app:frontdoor}.

\section{Proximal Causal Learning via Bridge Equations}\label{sec:2}

We now present the assumptions of the proximal causal learning approach via bridge equations, introduced by \citet{miao2018identifying} as a framework for addressing hidden confounding using confounder proxies, recovering the counterfactual distribution via the standard adjustment formula. \citet{ghassami24causal} develop analogous assumptions leveraging proxies of a hidden mediator, enabling identification via the frontdoor adjustment formula in Equation~\ref{eq:frontdoor}, Appendix~\ref{app:frontdoor}.

If variables \(C\) creating treatment-outcome  confounding in Figure~\ref{fig:1} are unobserved, the counterfactual distribution \(f_{Y(a)}(y)\), and average causal effects, are not identifiable without further assumptions. 

Consider the causal graphical model in Figure~\ref{fig:2}. In this figure, dashed edges (\tikz[baseline={(0,-0.5ex)}]{\draw[dashed,->] (0,0) -- (0.6,0);}) denote optional causal relationships; their presence or absence does not change the results discussed in this paper. The variables \(Z\) and \(W\) are descendants of, and reasonably associated with, the unmeasured confounder \(U\); we refer to them as \emph{proxies} of \(U\). Specifically, \(Z\) is a \emph{treatment-inducing} proxy, influencing \(A\), while \(W\) is an \emph{outcome-inducing} proxy, influencing \(Y\).

We emphasize that the term \emph{proxy} is used informally and lacks a universally consistent definition. Alternatively, under the bridge equation approach, \(Z\) may be a cause of \(U\), \(W\) may be a cause of \(U\), or \(A\) may be a cause of \(Z\). Admissible configurations of directed acyclic graphs for this approach are described in \citet{tchetgen2024proximal}.

Had \( U \) been observed, the adjustment formula would yield the identifying expression from Equation~\ref{eq:adjustment}, controlling for $U$ in place of $C$. \citet{miao2018identifying} replace the right hand side of Equation~\ref{eq:adjustment} with a functional of the observed law $f_{A,Y,W,Z}(a,y,w,z)$, leveraging Assumptions~\ref{assump:independence1} and~\ref{assump:bridge1}.

\begin{assumption}\label{assump:independence1}
    $W \perp\!\!\!\perp Z, A \mid U.$
\end{assumption}

\begin{assumption}\label{assump:independence2}
    $Z \perp\!\!\!\perp Y \mid U, A.$
\end{assumption}

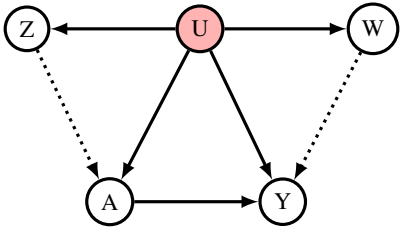
\begin{figure}[tb]
         \centering
        \begin{tikzpicture}[>=latex, node distance=2.3cm, anchor=center]
            \node (A) [draw, circle] {A};
            \node (Y) [draw, circle, right of=A] {Y};
             \node (U) [draw, circle, fill=red!30, above of=A, xshift=1.2cm] {U};
            \node (W) [draw, circle, right of=U] {W};
            \node (Z) [draw, circle, left of=U] {Z};

            \draw[->] (A) -- (Y);
            \draw[->] (U) -- (Z);
            \draw[->] (U) -- (W);
            \draw[->] (U) -- (Y);
            \draw[->] (U) -- (A);
            \draw[->, dashed] (Z) -- (A);
            \draw[->, dashed] (W) -- (Y);
        \end{tikzpicture}
        \caption{
    \centering
    Miao et al. (2018); \\
    Supports
    \(\mathcal{A}, \mathcal{Y}, \mathcal{U}, \mathcal{W}, \mathcal{Z}\)}
        \label{fig:2}
     \hfill
\end{figure}

Assumptions~\ref{assump:independence1} and~\ref{assump:independence2} are conditional independences reflected in the graphical model depicted in Figure~\ref{fig:2}. \citet{miao2018identifying} impose two additional conditions (Assumptions \ref{assump:completeness1} and~\ref{assump:bridge1}). We will elaborate upon these additional two assumptions and their meanings in Sections~\ref{subsec:2b} and~\ref{subsec:2c} below. A detailed proof of identification is given in their paper.  

\begin{assumption}[Completeness]\label{assump:completeness1}
    For each \(a \in \mathcal{A}\) and \(g\) square-integrable, \[\mathbb{E}[g(u) \mid z, a] = 0  \,\text{ a.s.} \iff g(u) = 0  \,\text{ a.s.}\]
\end{assumption}

\begin{assumption}[Solution to Bridge Equation]\label{assump:bridge1}
For each \(a \in \mathcal{A}\), there exists \(h(w,a,y)\) such that 
{\small
\[
f_{Y \mid Z,A}(y \mid z, a) 
= \int h(w, a, y)\, f_{W \mid Z,A}(w \mid z, a) \, dw < +\infty.
\]
}
\end{assumption}

It is important to emphasize that while this approach identifies the interventional distribution \(f_{Y(a)}(y)\), it does not, in general, identify functionals that depend on $U$—for instance, the conditional distribution \(f_{Y \mid U,A}(y \mid u, a)\) or the marginal distribution \(f_U(u)\) of the unobserved variable. This stands in contrast to the array decomposition approach discussed in Section~\ref{sec:3}.

\subsection{Completeness Conditions}\label{subsec:2b}

Assumption~\ref{assump:completeness1} is a generalized notion of linear independence in finite-dimensional vector spaces. To illustrate this idea, suppose for the time being that the model depicted in Figure~\ref{fig:2} is fully discrete. Then, Assumption~\ref{assump:completeness1} means that the conditional distributions in $\{f_{Z \mid U, A}(z \mid u, a): u \in \mathcal{U}\}$ are a finite set of linearly independent vectors for each \(a \in \mathcal{A}\). By Bayes' rule, this is equivalent to linear independence of distributional vectors in $\{f_{U \mid Z, A}(u \mid z, a): u \in \mathcal{U}\}$ for each \(a \in \mathcal{A}\). 

To make the connection with linear algebra explicit, 
define $\mathbf P_{U\mid Z,a}\in\mathbb R^{|\mathcal U|\times|\mathcal Z|}$ as in Appendix~\ref{app:matrix_def} with entries \(\big(\mathbf P_{U\mid Z,a}\big)_{ij}
\;:=\;
\Pr(U=u_i \mid Z=z_j, A=a).\)
Then we have

{\scriptsize
\begin{equation}\label{eq:expansion}
\begin{aligned}
&\begin{bmatrix}
\mathbb{E}[g(u) \mid z_1, a] \\ 
\mathbb{E}[g(u) \mid z_2, a] \\ 
\vdots \\
\mathbb{E}[g(u) \mid z_{|\mathcal{Z}|}, a]  
\end{bmatrix}
= 
\underbrace{
\begin{bmatrix}
\Pr(u_1 \mid z_1, a) & \cdots & \Pr(u_{|\mathcal{U}|} \mid z_1, a) \\ 
\Pr(u_1 \mid z_2, a) & \cdots & \Pr(u_{|\mathcal{U}|} \mid z_2, a) \\ 
\vdots & \ddots & \vdots \\ 
\Pr(u_1 \mid z_{|\mathcal{Z}|}, a) & \cdots & \Pr(u_{|\mathcal{U}|} \mid z_{|\mathcal{Z}|}, a) 
\end{bmatrix}
}_{\mathbf{P}_{U \mid Z, a}^T}
\begin{bmatrix}
g(u_1) \\
g(u_2) \\
\vdots \\
g(u_{|\mathcal{U}|})
\end{bmatrix}.
\end{aligned}
\end{equation}
}

Hence, Assumption~\ref{assump:completeness1} holds if \(\mathbf{P}_{U \mid Z,a}\) is right invertible. As noted in \citet{semiparametric2020cui}, \(Z\) must have at least as many categories as \(U\) in order for Assumption~\ref{assump:completeness1} to hold. 

Assumption~\ref{assump:completeness1} extends the notion of linear independence from finite-dimensional vector spaces to infinite-dimensional settings, where it is known as a \emph{completeness} condition. Completeness conditions are notoriously abstruse, but many commonly used statistical families are known to satisfy Assumption~\ref{assump:completeness1}. For example, the assumption holds when, for each \( a \in \mathcal{A} \), the conditional distributions 
\(
\{ f_{U \mid Z,A}(u \mid z, a) : z \in \mathcal{Z} \}
\)
form an exponential family indexed by \(z\), provided that \(z\) has support on an open interval of the real line \citep{Brown1986ExponentialFamilies, newey03instrumental}. At first, this formulation may appear more natural when \(Z\) is viewed as causing \(U\). However, if the family 
\(
\{ f_{U \mid Z,A}(u \mid z, a) : z \in \mathcal{Z} \}
\)
is specified as a regular exponential family indexed by \(z\), then the corresponding conditional distributions 
\(
\{ f_{Z \mid U,A}(z \mid u, a) : u \in \mathcal{U} \}
\)
belong to the associated conjugate exponential family under standard conjugacy relationships. 

\citet{Hu_Completeness} provide additional characterizations of nonparametric families of complete distributions, including settings in which \(U\) is generated as an arbitrary function of \(Z\) with independent additive noise. In general,  Assumption~\ref{assump:completeness1} is not verifiable \citep{Canay_testability2013}. However, as long as \(Z\) has larger dimension than \(U\), any distribution violating Assumption~\ref{assump:completeness1} can be approximated arbitrarily well by distributions that satisfy them \citep{semiparametric2020cui, Canay_testability2013}.

\subsection{Solution to a Bridge Equation}\label{subsec:2c}

The integral equation in Assumption~\ref{assump:bridge1}, referred to as a \emph{bridge equation}, enables recovery of the counterfactual distribution $f_{Y(a)}(y)$ alongside prior assumptions. Since the existence of a solution is usually assumed point blank, leaving researchers no way to reason about when this assumption is reasonable, we unpack sufficient conditions for and implications of this assumption.

In a fully discrete model, encode conditional distributions in matrices $\mathbf{P}_{Y \mid Z,a}, \mathbf{P}_{W \mid Z,a}, \mathbf{P}_{Y \mid U,a}$, and $\mathbf{P}_{W \mid U}$ (for details, see Appendix~\ref{app:matrix_def}). In this case, Assumption~\ref{assump:bridge1} is equivalent to the existence of a matrix \(\mathbf{H}\) such that
\begin{equation}\label{eq:discrete_bridge1}
\mathbf{P}_{Y \mid Z,a} \;=\; \mathbf{H}\,\mathbf{P}_{W \mid Z,a}.
\end{equation}

A new completeness condition, Assumption~\ref{assump:7a}, is sufficient for Assumption~\ref{assump:bridge1} in the discrete setting. 

\renewcommand{\theassumption}{7\alph{assumption}}
\setcounter{assumption}{0}

\begin{assumption}\label{assump:7a}
For each $a \in \mathcal{A}$ and every square-integrable function $g$,
\[
\mathbb{E}[g(Z) \mid W, A=a] = 0 \ \text{a.s.}
\quad \Longleftrightarrow \quad
g(Z) = 0 \ \text{a.s.}
\]
\end{assumption}

In the fully discrete case, this condition implies that $\mathbf{P}_{W \mid Z,a}$ admits a left inverse $\mathbf{P}_{W \mid Z,a}^{\dagger}$, so that
\begin{equation}\label{eq:discrete_bridge2}
\mathbf{P}_{Y \mid Z,a}
=
\mathbf{P}_{Y \mid Z,a}\,
\mathbf{P}_{W \mid Z,a}^{\dagger}\,
\mathbf{P}_{W \mid Z,a},
\end{equation}
which yields \(\mathbf{H}=\mathbf{P}_{Y \mid Z,a}\,
\mathbf{P}_{W \mid Z,a}^{\dagger}\) in Equation~\ref{eq:discrete_bridge1}. This left invertibility condition requires $W$ to have at least as many categories as $Z$.

Via the independences in Assumption~\ref{assump:independence1}, we can decompose Equation~\ref{eq:discrete_bridge2} as such:
\begin{equation}\label{eq:discrete_bridge3}
\mathbf{P}_{Y \mid U,a}\mathbf{P}_{U \mid Z,a}
=
\mathbf{P}_{Y \mid Z,a}\,
\mathbf{P}_{W \mid Z,a}^{\dagger}\,
\mathbf{P}_{W \mid U}\mathbf{P}_{U \mid Z,a}
\end{equation}

Under Assumption~\ref{assump:completeness1}, the matrix 
$\mathbf{P}_{U \mid Z,a}$ is right invertible. Consequently,
\[
\mathbf{P}_{Y \mid U,a}
=
\mathbf{P}_{Y \mid Z,a}\,
\mathbf{P}_{W \mid Z,a}^{\dagger}\,
\mathbf{P}_{W \mid U}.
\]

Taking the inner product of both sides with the marginal vector $f_U$, the left-hand side yields entries
\(
f_{Y(a)}(y) = \sum_u f_{Y \mid U,A}(y \mid u,a)\, f_U(u)
\)
and the right-hand side no longer relies on $u$ (since
$\sum_u f_{W \mid U}(w \mid u) f_U(u) = f_W(w)$).

In continuous models, Assumptions~\ref{assump:7b}--\ref{assump:7d} in addition to Assumption~\ref{assump:7a} are sufficient to ensure that Assumption~\ref{assump:bridge1} holds \citep{miao2018identifying}.

\begin{assumption}\label{assump:7b}
For each \(a \in \mathcal{A}\), \( \iint f_{W \mid Z,A}(w \mid z, a)\, f_{Z \mid W,A}(z \mid w, a) \, dw \, dz < +\infty. \)
\end{assumption}

\begin{assumption}\label{assump:7c}
For each \(a \in \mathcal{A}\),
\[\int f^2_{Y \mid Z,A}(y \mid z, a)\, f_{Z \mid A}(z \mid a) \, dz < +\infty.\]
\end{assumption}

\begin{assumption}\label{assump:7d}
For each \(a \in \mathcal{A}\), \[\sum_{n=1}^{\infty} \lambda_n^{-2} \left| \langle f_{Y \mid Z,A}(y \mid z, a), \psi_n \rangle \right|^2 < +\infty, \]
where \( (\lambda_n, \varphi_n, \psi_n)_{n=1}^{\infty}\) is the singular value decomposition of the operator \[\mathbf{K} : L^2\{F(w \mid a)\} \mapsto L^2\{F(z \mid a)\},\] defined by
\(\mathbf{K} h = \mathbb{E}\{h(w) \mid z, a\}\) for \( h \in L^2\{F(w \mid a)\} \). 
\end{assumption}

\renewcommand{\theassumption}{\arabic{assumption}}
\setcounter{assumption}{7}

Assumption~\ref{assump:7b} ensures that the operator \(\mathbf{K}\) is compact, exhibiting behavior similar to finite matrices. Assumption~\ref{assump:7c} places a square-integrability condition on \(f_{Y \mid Z,A}(y \mid z,a)\), and Assumption~\ref{assump:7d} controls the behavior of \(f_{Y \mid Z,A}(y \mid z,a)\) to lie in the range of \(\mathbf{K}\). Further discussion of these analytic conditions can be found in  \citet{kress1989}.

\section{
Causal Target Identification With Proxies via Array Decomposition}\label{sec:3}

Another prominent strategy for identifying counterfactual quantities with proxies of hidden variables employs array decomposition to recover the full law of hidden and observed variables, up to the labels of hidden variable values (e.g., \citet{kuroki14measurement, allman2015parameter, zhou2024causalinferencehiddentreatment}). We begin by reviewing this framework in the fully discrete hidden confounding model of \citet{kuroki14measurement}, consistent with the conditional independences depicted in Figure~\ref{fig:3}. In Appendix~\ref{app:KP_mediation}, we provide analogous assumptions for recovering causal effects in a discrete model with proxies of a hidden mediator. We defer discussion of continuous extensions to Section~\ref{subsec:3c}.

Although the graph in Figure~\ref{fig:3} does not impose discreteness, the array decomposition result from \citet{kuroki14measurement} is most transparently understood in a fully discrete version of this model. In this setting, with \(Z\) and \(W\) serving as proxies of the unobserved confounder \(U\), \citet{kuroki14measurement} show that the latent components \(f_{Y \mid U, A}(y \mid u_i, a)\) and \(f_U(u_i)\) are jointly identifiable up to a permutation of the latent categories of \(U\). Once these components are recovered, the counterfactual distribution \(f_{Y(a)}(y)\) may be derived from Equation~\ref{eq:adjustment}.

\begin{figure}[b]
         \centering
        \begin{tikzpicture}[>=latex, node distance=2.3cm, anchor=center]
            \node (A) [draw, circle] {A};
            \node (Y) [draw, circle, right of=A] {Y};
            \node (U) [draw, circle, fill=red!30, above of=A, xshift=1.2cm] {U};
            \node (W) [draw, circle, right of=U] {W};
            \node (Z) [draw, circle, left of=U] {Z};

            \draw[->] (A) -- (Y);
            \draw[->] (U) -- (Z);
            \draw[->] (U) -- (W);
            \draw[->] (U) -- (Y);
            \draw[->] (U) -- (A);
            \draw[->, dashed] (Z) -- (A);
        \end{tikzpicture}
        \caption{
    \centering
    Kuroki and Pearl (2014); \\
    Supports
    \(\mathcal{A}, \mathcal{Y}, \mathcal{U}, \mathcal{W}, \mathcal{Z}\)}
        \label{fig:3}
     \hfill
\end{figure}

\citet{kuroki14measurement} use the following assumptions, adopting the matrix notation $\mathbf{P}_{V \mid T}$ for encoding conditional distributions (see Appendix~\ref{app:matrix_def}),
which we use throughout the rest of the paper.

\begin{assumption}\label{assump:ind1}
\(W \perp\!\!\!\perp A \mid U.\)
\end{assumption}

\begin{assumption}\label{assump:ind2}
\[W, Z, Y \text{ are mutually independent 
conditional on }  U,A.\]
\end{assumption}

\begin{assumption}\label{assump:inv1}
For each \(a \in \mathcal{A}\), the matrices \(\mathbf{P}_{Z \mid U, a}\) and \(\mathbf{P}_{W \mid U}\) are invertible.
\end{assumption}

\begin{assumption}\label{assump:distinct_diag}
For \(a \in \mathcal{A}\) and \(y \in \mathcal{Y}\), $\Pr(Y = y \mid U = u_k, A = a)
\neq
\Pr(Y = y \mid U = u_{j}, A = a), \quad k \neq j.$
\end{assumption}

We provide a short summary of the identification proof given in \citet{kuroki14measurement} in Appendix~\ref{app:discrete_kp}. Assumption~\ref{assump:inv1} can be relaxed to left invertibility of \(\mathbf{P}_{Z \mid U, a}\) and \(\mathbf{P}_{W \mid U}\), allowing cases where \(Z\) and \(W\) each have at least as many levels as \(U\). Perhaps surprisingly, Assumption~\ref{assump:inv1} can even be relaxed to allow settings where both proxies \(Z\) and \(W\) have fewer levels than the latent variable. These generalizations are discussed in Section~\ref{subsec:3b}. 

Techniques for identifying causal targets via array decomposition, presented in works including \citet{kuroki14measurement}, \citet{allman2015parameter}, and  \citet{zhou2024causalinferencehiddentreatment} can be viewed as special cases of Kruskal’s uniqueness theorem for the Candecomp/Parafac (CP) decomposition \citep{kruskal1977three}, which we discuss in Section~\ref{subsec:3b}. As a consequence, any factor of the full law of hidden and observed variables can be recovered in these settings up to latent category labels. A natural extension of these ideas to continuous settings is explored in Section~\ref{subsec:3c}.

\subsection{Discrete Hidden Variable Models}\label{subsec:3b}

The array decomposition approach in discrete models is a special case of Kruskal’s uniqueness theorem for the Candecomp/Parafac (CP) decomposition of a three-way array (i.e., third-order tensor), as noted in \citet{allman2015parameter}. The theorem gives general conditions for uniqueness of a three-way array decomposition. Applied to the distribution of the observed variables, it provides conditions under which any factorization of the full law of hidden and observed variables is identified up to latent state labels, and rules out representations with a lower-dimensional latent variable. 

The remainder of this section presents the theorem and examines how its assumptions relate to those of Kuroki and Pearl (2014). A concise proof of the theorem is provided by \citet{rhodes20101818}, and the theorem is discussed at length in \citet{stegeman2007kruskal} and \citet{Derksen2013}.

\begin{definition}[CP Decomposition]\label{def:cp_decomp}
The rank-$R$ CP decomposition of a three-way array 
\(\boldsymbol{\mathcal{T}} \in \mathbb{R}^{M \times N \times J}\) 
consists of factor matrices 
\(\mathbf{A} \in \mathbb{R}^{M \times R}\), 
\(\mathbf{B} \in \mathbb{R}^{N \times R}\), and 
\(\mathbf{C} \in \mathbb{R}^{J \times R}\) such that, for each 
\(j = 1,\dots,J\), the matrix obtained by fixing the third index of 
\(\boldsymbol{\mathcal{T}}\) at \(j\),
\[
\boldsymbol{\mathcal{T}}_j := \boldsymbol{\mathcal{T}}(:,:,j),
\]
admits the representation \citep{stegeman2007kruskal, kruskal1977three}
\begin{equation}\label{eq:decomp}
\boldsymbol{\mathcal{T}}_j
=
\mathbf{A}\mathbf{C}_j\mathbf{B}^\top,
\end{equation}
where \(\mathbf{C}_j\) is a diagonal matrix whose diagonal entries are 
given by the \(j^{\text{th}}\) row of \(\mathbf{C}\).
\end{definition}

\begin{definition}[Essential Uniqueness] 
Three-way array \(\boldsymbol{\mathcal{T}}\) is essentially unique if its CP decomposition is uniquely determined up to permutation and scaling of the columns of \(\mathbf{A}\), \(\mathbf{B}\), and \(\mathbf{C}\).  
\end{definition}  

\begin{definition}[Three-Way Rank]\label{def:three_way_rank} 
The three-way rank of \(\boldsymbol{\mathcal{T}}\) is the smallest integer \(R\) for which a CP decomposition of the form in Equation~\eqref{eq:decomp} exists. 
\end{definition} 

\begin{definition}[Kruskal Rank ($k$-rank)]  
The Kruskal rank (or $k$-rank) of a matrix \(\mathbf{M}\), denoted \(k_M\), is the largest integer \(k\) such that any subset of \(k\) columns of \(\mathbf{M}\) is linearly independent.
\end{definition}  

\begin{theorem}[Kruskal's Uniqueness Theorem]\label{thm:kruskal}
A three-way array \(\boldsymbol{\mathcal{T}}\) has three-way rank R, and its rank-$R$ CP decomposition is essentially unique, if: 
\begin{equation}\label{eq:inequality}
2R + 2 \leq k_A + k_B + k_C,
\end{equation}
where \(k_A\), \(k_B\), and \(k_C\) are the $k$-ranks of \(\mathbf{A}\), \(\mathbf{B}\), and \(\mathbf{C}\), respectively.
\end{theorem}

Under the approach of Kuroki and Pearl (2014) described at the beginning of Section~\ref{sec:3}, let  
\[
\mathbf{A} = \mathbf{P}_{W \mid U}, \quad 
\mathbf{B} = \mathbf{P}_{Z, U \mid a} = \mathbf{P}_{Z \mid U,a}\,\operatorname{diag}(f_{U \mid a}),\]  
\[
\text{ and }\mathbf{C} = \mathbf{P}_{Y \mid U, a},
\]
where $\operatorname{diag}(f_{U \mid a})$ denotes the diagonal matrix with $(i,i)$-th entry $\Pr(U = u_i \mid A = a)$.
In Definition~\ref{def:cp_decomp}, this yields  
\[
\boldsymbol{\mathcal{T}}_j 
= \mathbf{A}\,\mathbf{C}_j\,\mathbf{B}^{\top} =
\mathbf{P}_{y_j, W, Z \mid, a} = \]
\[
\big[\,f_{Y,W,Z \mid A}(y_j, w, z_1 \mid a), \;\dots,\; f_{Y,W,Z \mid A}(y_j, w, z_{|\mathcal{Z}|} \mid a)\,\big].
\]

By Assumption~\ref{assump:inv1}, both \(\mathbf{A}\) and \(\mathbf{B}\) are invertible, which implies \(k_A = k_B = R = |\mathcal{U}|\). Furthermore, Assumption~\ref{assump:distinct_diag} ensures that \(k_C \geq 2\), since the columns of \(\mathbf{C}\) cannot be collinear; the conditional distributions \(f_{Y \mid U,A}(y \mid u, a)\) differ across values of latent variable \(U\). Hence \(2R + 2 \leq k_A + k_B + k_C\). 

Invoking Theorem~\ref{thm:kruskal}, identification of latent factors 
\[f_{W \mid U}(w \mid u_i), f_{Y \mid U,A}(y \mid u_i, a), \text{ and }f_{U,Z \mid A}(u_i, z \mid a)\] together reduces to solving the tensor decomposition separately for each \(a \in \mathcal{A}\). 
The decomposition identifies the latent categories only up to a permutation, so the index \(i\) serves merely as a label for the latent state \(u_i\). Since component \(f_{W \mid U}(w \mid u_i) = f_{W \mid U}(w \mid u_i, a )\) for all values of \(a \in \mathcal{A}\), labels can be aligned across treatment levels. Hence, \(f_{A,Y,U,W,Z}(a,y,u_i,w,z)\) is identified up to a permutation of the latent labels \(i\).

The original proof of Theorem~\ref{thm:kruskal} is nonconstructive \citep{kruskal1977three}, and can be recast in terms of uniquely determined eigenspaces \citep{rhodes20101818}. In the special setting considered by \citet{ kuroki14measurement}, the authors showed latent components can be obtained directly via eigendecomposition tasks. Modern methods for recovering latent factors under the general assumptions of Theorem~\ref{thm:kruskal} often rely on iterative procedures such as alternating least squares or gradient-based optimization \citep{ALS,gradient_descent}.

Notably, the inequality in Theorem~\ref{thm:kruskal} treats \(W\), \(Z\), and \(Y\) symmetrically. Consequently, the condition may still hold even if one observed variable (e.g., \(W\)) has fewer levels than the latent variable \(U\), provided that the other two variables (\(Z\) and \(Y\)) compensate. However, a necessary requirement for the inequality to hold is that \(W\), \(Z\), and \(Y\) together have at least \(2|\mathcal{U}| + 2\) distinct categories.

The CP decomposition framework is typically introduced in the simpler setting of a single latent variable with three observed variables that are mutually independent conditional on the latent (e.g., Figure~\ref{fig:3} omitting \(A\)). The three-way array formed by the observed variables uniquely determines the full law up to latent state labels under the condition in Theorem~\ref{thm:kruskal}. We now turn to how this framework extends to continuous models.

\subsection{Continuous Hidden Variable Models}\label{subsec:3c}

When the variables in Figure~\ref{fig:3} (except treatment \(A\)) are continuous, Theorem~1 of \citet{hu2008instrumental} provides a natural analogue of the array decomposition framework for uniquely determining the full law up to latent state labels. Originally developed for a setting with only three observed variables, \citet{deaner2023controllinglatentconfoundingtriple} extends the \citet{hu2008instrumental} result to the causal model in Figure~\ref{fig:3}, identifying the full law \(f_{A,Y,U,W,Z}(a,y,u_i,w,z)\) and thereby counterfactual distribution \(f_{Y(a)}(y)\).

The extension relies on the conditional independences in Assumptions~\ref{assump:ind1}–\ref{assump:ind2}, the regularity condition in Assumption~\ref{assump:KP_bounded}, and additional conditions Assumptions~\ref{assump:completeness4} and~\ref{assump:noncollinearity2}.

\begin{assumption}\label{assump:KP_bounded}
\((W, Z, Y, U, A)\) admit a joint density that is bounded and nonzero.
\end{assumption}

\begin{assumption}\label{assump:completeness4}
    For each \(a \in \mathcal{A}\) and \(g\) square-integrable, \(\mathbb{E}[g(u) \mid w, a] = 0  \,\text{ a.s.} \iff g(u) = 0  \,\text{ a.s.}\) and 
    \(\mathbb{E}[g(u) \mid z, a] = 0  \,\text{ a.s.} \iff g(u) = 0  \,\text{ a.s.}\)
\end{assumption}

\begin{assumption}\label{assump:noncollinearity2}
For each \(a \in \mathcal{A}\) and \(i \neq j\), the set
\[
\{y: f_{Y \mid U,A}(y \mid u_i, a) \neq f_{Y \mid U,A}(y \mid u_j, a)\}
\]
has positive probability under the marginal of $Y \mid A=a$.
\end{assumption}

Recalling the relationship between completeness and invertibility discussed in Section~\ref{subsec:2b}, Assumption~\ref{assump:completeness4} is analogous to Assumption~\ref{assump:inv1} in \citet{kuroki14measurement}. Additionally, Assumption~\ref{assump:noncollinearity2} is analogous to Assumption~\ref{assump:distinct_diag}. Thus, the results of \citet{hu2008instrumental} and \citet{deaner2023controllinglatentconfoundingtriple} can be viewed as natural extensions of \citet{kuroki14measurement} to the continuous setting.

\subsection{Identifying Latent Labels}\label{subsec:3d}
In the hidden confounding setting considered in Figure~\ref{fig:3}, label recovery for the latent confounder \(U\) is not required for identification of counterfactual distributions or causal effects. The confounder is not the target of intervention; it is integrated out in the adjustment formula (Equation~\ref{eq:adjustment}, with $U$ in place of $C$). Any relabeling of the latent categories of \(U\) leaves the functional (i.e.,
\(
\int f_{Y \mid A,U}(y \mid a, u)\, f_U(u)\, du
\))
unchanged.

In contrast, outside of the hidden confounding context, it is often necessary to impose additional conditions in order to recover the labels of the latent variable. For example, when treatment \(A\) itself is unobserved but proxies of it are available, allowing identification of $f_{Y \mid A,C}(y \mid a_i, c)$ \citep{zhou2024causalinferencehiddentreatment}, recovering which $i$ corresponds to \(a_1\) versus \(a_0\) is essential to identify a specific counterfactual distribution
\[
f_{Y(a_i)}(y) 
= \int f_{Y \mid A,C}(y \mid a_i, c)\, f_C(c)\, dc,
\]
because intervention is defined relative to a particular treatment level.

\subsubsection{Conditional Unbiasedness}\label{subsec:3di}

\citet{hu2008instrumental} show that in the model of Figure~\ref{fig:3}, under Assumptions~\ref{assump:ind1}--\ref{assump:ind2} and~\ref{assump:KP_bounded}--\ref{assump:noncollinearity2}, latent state labels can be recovered under an unbiasedness condition: a known functional of an observed variable (e.g., a measure of central tendency) evaluated conditional on each latent value equals the latent value itself (Assumption~\ref{assump:unbiasedness}).

\begin{assumption}\label{assump:unbiasedness}  
There exists a known functional \( M \) such that for each \( u \in \mathcal{U} \),
\[
M[f_{W \mid U}(w \mid u)] = u .
\]
\end{assumption}

\subsubsection{Conditional Monotonicity}\label{subsec:3dii}

When \(U\) is ordinal and its support known, an alternative condition, Assumption~\ref{assump:monotonicity2}, suffices. Assumption~\ref{assump:monotonicity2} requires that a known functional of an observed variable is strictly ordered across latent categories; since the full joint law is identified, this ordering determines the labels.

\begin{assumption}\label{assump:monotonicity2}
\(U\) is ordinal with known support and there exists a known functional \( M \) such that, for all \( u_i < u_j \)
\[
M[f_{W \mid U}(w \mid u_i)] < M[f_{W \mid U}(w \mid u_j)].\]
\end{assumption}

In both Assumption~\ref{assump:unbiasedness} and~\ref{assump:monotonicity2}, \(W\) may be replaced by \(Z\) or \(Y\), provided the condition holds additionally conditional on each \(a \in \mathcal{A}\).

\section{Model Comparison}\label{sec:4}

The array decomposition approach in Section~\ref{sec:3} can be used to recover counterfactual quantities across a broad class of hidden variable models. For example, \citet{zhou2024causalinferencehiddentreatment} apply the method to hidden treatment models. This flexibility underscores a key strength of the array decomposition framework. Moreover, the two proximal strategies for hidden confounding invoke distinct conditional independence restrictions, which are non-nested \citep{deaner2023controllinglatentconfoundingtriple}. 

We now compare the two approaches under shared conditional independences compatible with both frameworks. Specifically, we consider the hidden confounding model represented in Figure~\ref{fig:3}, i.e., reflecting the conditional independences in Assumptions~\ref{assump:ind1} and~\ref{assump:ind2}.

\begin{table*}[tb]
    \centering
    \caption{Comparison for discrete hidden confounding model in Figure~\ref{fig:3}.$^{\ast}$}
    \label{tab:discrete}
    \renewcommand{\arraystretch}{1.35}
    \setlength{\tabcolsep}{6pt}
    \small
    \begin{tabularx}{\linewidth}{|X|X|X|}
        \hline
        \textbf{Bridge Equation Approach$^{\dagger}$} & 
        \textbf{Array Decomposition Approach$^{\ddagger}$} & 
        \textbf{Overlap} \\
        \hline
        \multicolumn{3}{|c|}{\textbf{Assumption Set 1}} \\
        \hline
        \begin{titemize}[leftmargin=*,nosep]
            \item $\mathbf{P}_{Z \mid U, a}$ is left invertible
            \item $\exists \mathbf{H}$ s.t. 
            $\mathbf{P}_{Y \mid Z,a} = \mathbf{H}\,\mathbf{P}_{W \mid Z,a}$
        \end{titemize}
        &
        \begin{titemize}[leftmargin=*,nosep]
            \item $\mathbf{P}_{Z \mid U, a}^{-1}$ and $\mathbf{P}_{W \mid U}^{-1}$ exist
            \item Some row of $\mathbf{P}_{Y \mid U,a}$ has distinct entries
        \end{titemize}
        &
        \begin{titemize}[leftmargin=*,nosep]
            \item $\mathbf{P}_{Z \mid U, a}^{-1}$ and $\mathbf{P}_{W \mid U}^{-1}$ exist
            \item Some row of $\mathbf{P}_{Y \mid U,a}$ has distinct entries
        \end{titemize}
        \\
        \hline
        \multicolumn{3}{|c|}{\textbf{Assumption Set 2}} \\
        \hline
        \begin{titemize}[leftmargin=*,nosep]
            \item $\mathbf{P}_{Z \mid U, a}$ is left invertible
            \item $\exists \mathbf{H}$ s.t. 
            $\mathbf{P}_{Y \mid Z,a} = \mathbf{H}\,\mathbf{P}_{W \mid Z,a}$
        \end{titemize}
        &
        \begin{titemize}[leftmargin=*,nosep]
            \item $2|\mathcal{U}| + 2 \leq k_A + k_B + k_C$
        \end{titemize}
        &
        \begin{titemize}[leftmargin=*,nosep]
            \item $\mathbf{P}_{Z \mid U, a}$ is left invertible 
            \newline $\implies k_B = |\mathcal{U}|$
            \item $\exists \mathbf{H}$ s.t. 
            $\mathbf{P}_{Y \mid Z,a} = \mathbf{H}\,\mathbf{P}_{W \mid Z,a}$
            \item $|\mathcal{U}| + 2 \leq k_A + k_C$
        \end{titemize}
        \\
        \hline
        \multicolumn{3}{|c|}{\textbf{Identified Targets}} \\
        \hline
        \begin{titemize}[leftmargin=*,nosep]
            \item $f_{Y(a)}(y)$
        \end{titemize}
        &
        \begin{titemize}[leftmargin=*,nosep]
            \item $f_{A,Y,U,W,Z}(a,y,u,w,z)$ up to latent labels
            \item $f_{Y(a)}(y)$
        \end{titemize}
        &
        \begin{titemize}[leftmargin=*,nosep]
            \item $f_{A,Y,U,W,Z}(a,y,u,w,z)$ up to latent labels
            \item $f_{Y(a)}(y)$
        \end{titemize}
        \\
        \hline
    \end{tabularx}

    \vspace{0.4em}
    \begin{minipage}{0.95\linewidth}
    \footnotesize
    $^{\ast}$ All assumptions are stated for each \(a \in \mathcal{A}\). \\
    $^{\dagger}$ Bridge equation approach assumptions are from \citet{miao2018identifying}. \\
    $^{\ddagger}$ Array decomposition approach assumptions are as follows:  
    Set 1 from \citet{kuroki14measurement};  
    Set 2 from \citet{kruskal1977three}. 
    \end{minipage}
\end{table*}

\subsection{Discrete Model Comparison}
For discrete models, we compare Assumption~\ref{assump:completeness1} and~\ref{assump:bridge1} from the bridge equation approach in Miao et al.~(2018) to two alternative assumption sets (Table~\ref{tab:discrete}). Set 1 is from \citet{kuroki14measurement}; Set 2 is from 
\citet{kruskal1977three}, the generalization of Set 1 via Theorem~\ref{thm:kruskal}.

Comparing to Set 1, all model overlap occurs under the assumptions of \citet{kuroki14measurement}. In particular, the invertibility conditions in Assumption~\ref{assump:inv1} from \citet{kuroki14measurement}---namely, that $\mathbf{P}_{Z \mid U,a}^{-1}$ and $\mathbf{P}_{W \mid U}^{-1}$ exist---also imply the existence of $\mathbf{P}_{W \mid Z,a}^{-1}$. This directly satisfies Assumptions~\ref{assump:completeness1} and~\ref{assump:bridge1} in \citet{miao2018identifying}, since one can define  
\[
\mathbf{H} = \mathbf{P}_{Y \mid Z,a}\,\mathbf{P}_{W \mid Z,a}^{-1}
\quad \text{so that} \quad
\mathbf{P}_{Y \mid Z,a} = \mathbf{H}\,\mathbf{P}_{W \mid Z,a}.
\]
Thus, the model from \citet{kuroki14measurement} is contained within that of \citet{miao2018identifying}. In other words, \citet{kuroki14measurement} imposes stronger assumptions.

Comparing to Set 2, again let
\[
\mathbf{A} = \mathbf{P}_{W \mid U}, \quad 
\mathbf{B} = \mathbf{P}_{Z \mid U,a}\,\operatorname{diag}(f_{U \mid a}), \quad 
\mathbf{C} = \mathbf{P}_{Y \mid U, a},
\]
in Definition~\ref{def:cp_decomp}, so that  
\(
\boldsymbol{\mathcal{T}}_j 
= \mathbf{A}\,\mathbf{C}_j\,\mathbf{B}^{\top} = \mathbf{P}_{y_j, W, Z \mid, a} =\)
\[
\big[\,f_{Y,W,Z \mid A}(y_j, w, z_1 \mid a), \;\dots,\; f_{Y,W,Z \mid A}(y_j, w, z_{|\mathcal{Z}|} \mid a)\,\big].
\] Define \(k_A\), \(k_B\), and \(k_C\) as the $k$-ranks of \(\mathbf{A}\), \(\mathbf{B}\), and \(\mathbf{C}\), respectively. 

Assumptions~\ref{assump:completeness1} 
and~\ref{assump:bridge1} from \citet{miao2018identifying} do not imply that $2|\mathcal{U}| + 2 \leq k_A + k_B + k_C$, nor does the converse hold. In other words, the models are non-nested.

\subsection{Continuous Model Comparison}

Table~\ref{tab:continuous} reviews the assumptions of \citet{miao2018identifying} alongside those of \citet{hu2008instrumental} for continuous models. We conjecture that the two models are non-nested in the model represented by Figure~\ref{fig:3}. We give an example satisfying the assumptions of both approaches, which appears in the latter half of \citet{kuroki14measurement} and is revisited in \citet{miao2018identifying}.

\begin{table*}[tb]
    \centering
    \caption{Comparison for continuous hidden confounding model in Figure~\ref{fig:3}.$^{\ast}$}
    \label{tab:continuous}
    \renewcommand{\arraystretch}{1.35}
    \setlength{\tabcolsep}{6pt}
    \small
    \begin{tabularx}{\linewidth}{|X|X|X|}
        \hline
        \textbf{Bridge Equation Approach$^{\dagger}$} &
        \textbf{Array Decomposition Approach$^{\ddagger}$} &
        \textbf{Overlap Example} \\
        \hline
        \multicolumn{3}{|c|}{\textbf{Assumptions}} \\
        \hline
        \begin{titemize}[leftmargin=*,nosep]
            \item \textbf{Completeness}: 
            \newline $\text{For } g \text{ square-integrable, }$ \newline $\mathbb{E}[g(u) \mid z,a]=0 \;\mathrm{ a.s.} \iff g(u)=0 \;\mathrm{ a.s.}$
            \vspace{0.3em}
            \item \textbf{Bridge equation}: 
            \newline $\exists\, h(w,a,y)$ s.t. 
            $f_{Y \mid Z,A}(y \mid z,a)=\int h(w,a,y)\, f_{W \mid Z,A}(w \mid z,a)\, dw < \infty$
        \end{titemize}
        &
        \begin{titemize}[leftmargin=*,nosep]
            \item \textbf{Bounded densities}: 
            \newline {\small \((W, Z, Y, U, A)\) admit a joint density that is bounded and nonzero.}
            \vspace{0.3em}
            \item \textbf{Completeness}: 
            \newline $\text{For } g \text{ square-integrable, }$ \newline $\mathbb{E}[g(u) \mid z,a]=0 \;\mathrm{ a.s.} \iff g(u)=0 \;\mathrm{ a.s.}$ and $\mathbb{E}[g(u) \mid w,a]=0 \;\mathrm{ a.s.} \iff g(u)=0 \;\mathrm{ a.s.}$
            \vspace{0.3em}
            \item \textbf{Non-collinearity}: 
            \newline The set \(\{y: f_{Y \mid U,A}(y \mid u_i, a) \neq f_{Y \mid U,A}(y \mid u_j, a)\}\) has positive probability under the marginal of $Y \mid A=a$
        \end{titemize}
        &
        \begin{titemize}[leftmargin=*,nosep]
            \item[] \[\hspace{0pt}\begin{aligned}
                    U &= \mu_U + \varepsilon_U,\\
                    Z &= \beta_{0Z} + \alpha_{UZ}U + \varepsilon_Z,\\
                    A &= \beta_{0A} + \alpha_{UA}U + \alpha_{ZA}Z + \varepsilon_A,\\
                    W &= \beta_{0W} + \alpha_{UW}U + \varepsilon_W,\\
                    Y &= \beta_{0Y} + \alpha_{AY}A + \alpha_{UY}U + \varepsilon_Y,
                    \end{aligned}
                    \]
                    with all disturbances mean-zero Gaussian and mutually independent
        \end{titemize}
        \\
        \hline
        \multicolumn{3}{|c|}{\textbf{Identified Targets}} \\
        \hline
        \begin{titemize}[leftmargin=*,nosep]
            \item $f_{Y(a)}(y)$
        \end{titemize}
        &
        \begin{titemize}[leftmargin=*,nosep]
            \item $f_{A,Y,U,W,Z}(a,y,u,w,z)$ up to latent labels
            \vspace{0.3em}
            \item $f_{Y(a)}(y)$
        \end{titemize}
        &
        \begin{titemize}[leftmargin=*,nosep]
            \item $f_{A,Y,U,W,Z}(a,y,u,w,z)$ up to latent labels
            \vspace{0.3em}
            \item $f_{Y(a)}(y)$
        \end{titemize}
        \\
        \hline
    \end{tabularx}

    \vspace{0.4em}
    \begin{minipage}{0.95\linewidth}
    \footnotesize
    $^{\ast}$ All assumptions are stated for each \(a \in \mathcal{A}\). \\
    $^{\dagger}$ Bridge equation approach assumptions are from \citet{miao2018identifying}. \\
    $^{\ddagger}$ Array decomposition approach assumptions are from \citet{hu2008instrumental}, from which \citet{deaner2023controllinglatentconfoundingtriple} shows \(f_{Y(a)}(y)\) is identified.
    \end{minipage}
\end{table*}

\section{A Brief Note on Estimation}\label{sec:5}

While our primary focus in this paper is on comparing model assumptions for identification, it is worth noting that \citet{miao2018identifying} provide an illustrative example in which \(f_{Y(a)}(y)\) can be derived analytically under location--scale family assumptions, yielding a tractable framework for consistent estimation of causal effects. In their work, the efficient influence function has also been established under further model assumptions, giving rise to estimation of causal effects achieving semiparametric efficiency \citep{semiparametric2020cui}. In contrast, \citet{hu2008instrumental} rely on sieve methods to approximate factors of the full law. Some semiparametric theory has also been developed for estimating causal effects recovered from the array decomposition approach in the context of proxies of hidden treatment, though this literature remains relatively nascent \citep{zhou2024causalinferencehiddentreatment}.

\section{Conclusion and Future Directions}\label{sec:6}

Proxy-based approaches have emerged as a promising solution to identifying causal targets when hidden variables pose an obstacle. Within this framework, two complementary strategies stand out. The bridge equation method identifies causal targets by solving functional equations, while the array decomposition approach leverages unique determination of eigenspaces to uncover latent distributions which constitute those targets. Although both methods hinge on assumptions regarding the informativeness of proxy variables for the latent structure, comparing their respective assumptions and areas of overlap offers clearer insight into the contexts in which each method is most appropriately applied.

Looking ahead, several directions for future research present themselves. One avenue is to explore how proximal methods can extend to more complicated scenarios that are generally non-identifiable. A further direction is to establish general semiparametric theory and influence function-based estimators for causal effects identified through these proximal approaches. Such advances would expand the scope of settings in which causal targets are identifiable and provide robust tools for researchers working with complex data to estimate them.

\appendix

\section{Definitions}\label{app:matrix_def}
Let $V$ and $T$ be
discrete random variables with supports
\[
\mathcal V=\{v_1,\dots,v_{|\mathcal V|}\}, 
\qquad
\mathcal T=\{t_1,\dots,t_{|\mathcal T|}\}.
\]

For each $a\in\mathcal A$, define matrix
$\mathbf P_{V\mid T,a}\in\mathbb R^{|\mathcal V|\times|\mathcal T|}$ by
\begin{equation}\label{eq:matrix}
\big(\mathbf P_{V\mid T,a}\big)_{ij}
:=
\Pr(V=v_i \mid T=t_j, A=a).
\end{equation}

When conditioning on $A$ is not intended, we write matrix
$\mathbf P_{V\mid T}$, defined analogously by
\[
\big(\mathbf P_{V\mid T}\big)_{ij}
:=
\Pr(V=v_i \mid T=t_j).
\]
\section{Frontdoor Adjustment Formula}\label{app:frontdoor}

Consider the hidden variable model in Figure~\ref{fig:A1}, where variables \(U\) creating treatment-outcome confounding are unobserved, however a mediator \(M\) between treatment \(A\) and outcome \(Y\) is observed. Here, the counterfactual distribution \(f_{Y(a)}(y)\) is identified by the frontdoor formula,
{\small
\begin{equation}\label{eq:frontdoor}
f_{Y(a)}(y) =  
\iint f_{Y \mid A,M}(y \mid a', m) \, f_{M \mid A}(m \mid a) \, f_A(a') \, da' \, dm,
\end{equation}} provided the following conditions hold for each \(a \in \mathcal{A}\) and \(m \in \mathcal{M}\):  

\begin{assumption}[Counterfactual Independences] 
\[M(a) \perp\!\!\!\perp A \text{ and } Y(m) \perp\!\!\!\perp M \mid A,\]
\end{assumption}

\begin{assumption}[Exclusion restriction]
\[Y(a,m) = Y(m),\]
\end{assumption}

\begin{assumption}[Consistency]
\[M = M(a) \text{ when } A = a, \text{ and }\]
\[Y = Y(a,m) \text{ when } A = a, M = m,\]
\end{assumption}

\begin{assumption}[Positivity]
\[f_{M \mid A}(m \mid a) > 0 \text{ and } f_A(a) > 0.\]
\end{assumption}

\begin{figure}[b]
    \centering
    \begin{tikzpicture}[>=latex, node distance=2.3cm, anchor=center]
        \node (A) [draw, circle] {A};
        \node (M) [draw, circle, right of=A] {M};
        \node (Y) [draw, circle, right of=M] {Y};
        \node (U) [draw, circle, fill=red!30, above of=A, xshift=2.2cm] {U};

        \draw[->] (A) -- (M);
        \draw[->] (M) -- (Y);
        \draw[->] (U) -- (A);
        \draw[->] (U) -- (Y);
    \end{tikzpicture}
    \caption{\centering 
    Hidden confounding with observed mediation; \\
    Supports $\mathcal{A,M,Y,U}$} 
    \label{fig:A1}
\end{figure}
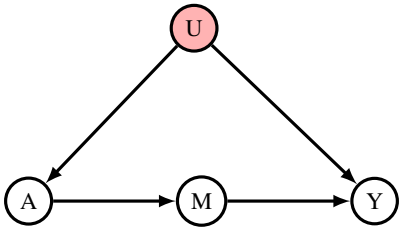

\section{Kuroki and Pearl (2014)}\label{app:discrete_kp}
We summarize the identification proof from \citet{kuroki14measurement}.

\begin{proof}

By Assumption~\ref{assump:ind1} and Assumption~\ref{assump:ind2}, for each \(a \in \mathcal{A}\) and \(y \in \mathcal{Y}\),
\begin{equation}
    \mathbf{P}_{W \mid Z, a} = \mathbf{P}_{W \mid U} \mathbf{P}_{U \mid Z, a}, \text{ and} \\
\end{equation}
\begin{equation}
    \mathbf{P}_{y, W \mid Z, a} = \mathbf{P}_{W \mid U} \operatorname{diag}(f_{y \mid U, a}) \mathbf{P}_{U \mid Z, a},
\end{equation}
where $\operatorname{diag}(f_{y \mid U, a})$ denotes the diagonal matrix with $(i,i)$-th entry $\Pr(Y = y \mid U = u_i, A = a)$.

By Assumption~\ref{assump:inv1}, \(\mathbf{P}_{U \mid Z, a}^{-1}\) exists and  
\[
\mathbf{P}_{W \mid Z, a}^{-1} = \mathbf{P}_{U \mid Z, a}^{-1}\, \mathbf{P}_{W \mid U}^{-1}.
\] Hence,
\begin{equation}
\mathbf{P}_{y, W \mid Z, a} \mathbf{P}_{W \mid Z, a}^{-1} = \mathbf{P}_{W \mid U} \operatorname{diag}(\mathbf{P}_{y \mid U, a})  \mathbf{P}_{W \mid U}^{-1}.
\end{equation}

Then eigendecomposition allows us to identify the column vectors of 
\(\mathbf{P}_{W \mid U}\) under Assumption~\ref{assump:distinct_diag}. In other words, we recover conditional distributions \(f_{W \mid U}(w \mid u_i)\), up to the labeling of the latent variable \(U\).  
Inverting   
\(
f_{Y,W \mid A}(y, w \mid a) = \sum_{i} f_{Y,U \mid A}(y, u_i \mid a)\, f_{W \mid U}(w \mid u_i),
\)  
we recover \[f_{Y\mid U, A}(y \mid u_i, a), \quad f_{U \mid A}(u_i \mid a).\] Further, the component \(f_{W \mid U}(w \mid u_i) = f_{W \mid U}(w \mid u_i, a)\) for all values \(a \in \mathcal{A}\), so that the labels can be aligned across treatment levels, yielding the components in Equation~\eqref{eq:adjustment}.

\end{proof}

\section{Mediator Proxy Analogue for Kuroki and Pearl (2014)}\label{app:KP_mediation}

We provide analogous assumptions to Assumption~\ref{assump:ind1}-\ref{assump:distinct_diag} in \citet{kuroki14measurement}, instead leveraging proxies of a hidden mediator depicted in Figure~\ref{fig:A3}. 

\begin{assumption}\label{assump:mediation_ind1}
\(W \perp\!\!\!\perp A \mid M.\)
\end{assumption}

\begin{assumption}\label{assump:mediation_ind2}
\[W, Z, Y \text{ are mutually independent 
conditional on } M,A.\]
\end{assumption}

\begin{assumption}\label{assump:mediation_inv1}
For each \(a \in \mathcal{A}\), the matrices \(\mathbf{P}_{Z \mid M, a}\) and \(\mathbf{P}_{W \mid M}\) are invertible.
\end{assumption}

\begin{assumption}\label{assump:mediation_distinct_diag}
For each \(a \in \mathcal{A}\), some row of $\mathbf{P}_{Y \mid M,a}$ has distinct entries.
\end{assumption}

\begin{figure}[b]
    \centering
    \begin{tikzpicture}[>=latex, node distance=2.3cm, anchor=center]
        \node (A) [draw, circle] {A};
        \node (M) [draw, circle, fill=red!30, right of=A] {M};
        \node (Y) [draw, circle, right of=M] {Y};
        \node (W) [draw, circle, above of=Y] {W};
        \node (Z) [draw, circle, above of=A] {Z};
        \node (U) [draw, circle, fill=red!30, above of=M] {U};

        \draw[->] (A) -- (M);
        \draw[->] (M) -- (Y);
        \draw[->] (U) -- (Y);
        \draw[->] (U) -- (A);
        \draw[->] (M) -- (W);
        \draw[->] (M) -- (Z);
        \draw[->, dotted] (A) -- (Z);
    \end{tikzpicture}
    \caption{ \centering Analogue of Kuroki and Pearl (2014);\\
    Supports $\mathcal{A,M,Y,U,W,Z}$}
    \label{fig:A3}
\end{figure}

$f_{Y(a)}(y)$ is identified by the proof in Appendix~\ref{app:discrete_kp}, replacing \(U\) with \(M\).

\begin{acks}[Acknowledgments]
The authors are grateful to Eric Tchetgen Tchetgen, Trung Phung, and Beatrix Wen for helpful discussions.
\end{acks}


\begin{funding}
ONR Grant N000142412701.
\end{funding}


\bibliographystyle{apalike}
\bibliography{references}

@book{kress1989,
  author    = {R. Kress},
  title     = {Linear Integral Equations},
  year      = {1989},
  publisher = {Springer},
  address   = {Berlin}
}

@misc{deaner2023controllinglatentconfoundingtriple,
      title={Controlling for Latent Confounding with Triple Proxies}, 
      author={Ben Deaner},
      year={2023},
      eprint={2204.13815},
      archivePrefix={arXiv},
      primaryClass={econ.EM},
      url={https://arxiv.org/abs/2204.13815}, 
}

@article{zhou2024causalinferencehiddentreatment,
      title={Causal Inference for a Hidden Treatment}, 
      author={Y. Zhou and E. {Tchetgen Tchetgen}},
      year={2024},
      eprint={2405.09080},
      archivePrefix={arXiv},
      primaryClass={stat.ME},
      url={https://arxiv.org/abs/2405.09080}, 
      journal = { arXiv preprint:
		\url{https://arxiv.org/abs/2405.09080}},
}

@article{allman2015parameter,
  author    = {E. Allman and J. Rhodes and E. Stanghellini and M. Valtorta},
  title     = {Parameter Identifiability of Discrete Bayesian Networks with Hidden Variables},
  journal   = {Journal of Causal Inference},
  year      = {2015},
  volume    = {3},
  number    = {2},
  pages     = {189--205},
  doi       = {10.1515/jci-2014-0021},
  url       = {https://doi.org/10.1515/jci-2014-0021}
}

@article{stegeman2007kruskal,
  author    = {A. Stegeman and N. Sidiropoulos},
  title     = {On Kruskal’s uniqueness condition for the Candecomp/Parafac decomposition},
  journal   = {Linear Algebra and its Applications},
  year      = {2007},
  volume    = {420},
  number    = {2--3},
  pages     = {540--552},
  issn      = {0024-3795},
  doi       = {10.1016/j.laa.2006.08.010},
  url       = {https://doi.org/10.1016/j.laa.2006.08.010}
}

@article{kruskal1977three,
  author    = {J. Kruskal},
  title     = {Three-way arrays: rank and uniqueness of trilinear decompositions, with applications to arithmetic complexity and statistics},
  journal   = {Linear Algebra and its Applications},
  year      = {1977},
  volume    = {18},
  pages     = {95--138},
  doi       = {10.1016/0024-3795(77)90069-6},
  url       = {https://doi.org/10.1016/0024-3795(77)90069-6}
}

@article{hu2008instrumental,
  author  = {Yingyao Hu and Susanne Schennach},
  title   = {Instrumental Variable Treatment of Nonclassical Measurement Error Models},
  journal = {Econometrica},
  volume  = {76},
  number  = {1},
  pages   = {195--216},
  year    = {2008},
  doi     = {10.1111/j.0012-9682.2008.00823.x},
  url     = {https://doi.org/10.1111/j.0012-9682.2008.00823.x}
}

@article{newey03instrumental,
    title = { Instrumental variable estimation of nonparametric models },
    author = { Newey, W. K. and Powell, J. L. },
    journal = { Econometrica },
    volume = { 71 },
    pages = { 1565-1578 },
    year = { 2003 }
}

@article{ghassami24causal,
  author  = {Ghassami, A. and Yang, A. and Shpitser, I. and Tchetgen Tchetgen, E.},
  title   = {Causal inference with hidden mediators},
  journal = {Biometrika},
  year    = {2025},
  volume  = {112},
  number  = {1},
  pages   = {asae037},
  doi     = {10.1093/biomet/asae037}
}

@article{semiparametric2020cui,
author = {Yifan Cui and Hongming Pu and Xu Shi and Wang Miao and Eric Tchetgen Tchetgen},
title = {Semiparametric Proximal Causal Inference},
journal = {Journal of the American Statistical Association},
volume = {119},
number = {546},
pages = {1348--1359},
year = {2024},
publisher = {Taylor \& Francis},
doi = {10.1080/01621459.2023.2191817},
URL = { 
    
        https://doi.org/10.1080/01621459.2023.2191817

},
}

@article{miao2018identifying,
  title={Identifying causal effects with proxy variables of an unmeasured confounder},
  author={Miao, Wang and Geng, Zhi and {Tchetgen Tchetgen}, Eric J},
  journal={Biometrika},
  volume={105},
  number={4},
  pages={987--993},
  year={2018},
  publisher={Oxford University Press}
}

@article{vanderweele13on,
	author = { Tyler J. Vander{W}eele and Ilya Shpitser },
   	title = { On the definition of a confounder },
   	journal = { Annals of Statistics },
   	volume = { 41 },
   	number = { 1 },
   	pages = { 196-220 },
   	year = { 2013 }}

@article{richardson13swig,
	author = { Thomas S. Richardson and Jamie M. Robins },
	title = { Single World Intervention Graphs ({SWIG}s):
		A Unification of the Counterfactual and Graphical
		Approaches to Causality },
	journal = { preprint:
		\url{http://www.csss.washington.edu/Papers/wp128.pdf}},
	year = { 2013 }}

@article{kuroki14measurement,
	author = { Manabu Kuroki and Judea Pearl },
	title = { Measurement bias and effect restoration in causal inference },
	year = { 2014 },
	journal = { Biometrika },
	volume = { 101 },
	issue = { 2 },
	pages = { 423-437 }}

@article{pearl95causal,
	Author = {Judea Pearl},
	Journal = {Biometrika},
	Number = {4},
	Pages = {669--709},
	Title = {Causal Diagrams for Empirical Research},
	Url = {citeseer.ist.psu.edu/55450.html},
	Volume = {82},
	Year = {1995}}

@book{pearl88probabilistic,
	Author = {Judea Pearl},
	Publisher = {Morgan and Kaufmann, San Mateo},
	Title = {Probabilistic Reasoning in Intelligent Systems},
	Year = {1988}}

@book{pearl09causality,
	Author = {Judea Pearl},
	Isbn = {978-0521895606},
	Publisher = {Cambridge University Press},
	edition = { 2 },
	Title = {Causality: Models, Reasoning, and Inference},
	Year = {2009}}

@book{spirtes01causation,
	Author = {Peter Spirtes and Clark Glymour and Richard Scheines},
	Publisher = {Springer Verlag, New York},
	Title = {Causation, Prediction, and Search},
	isbn = { 978-0262194402 },
	edition = { 2 },
	Year = {2001}}

@article{rhodes20101818,
title = {A concise proof of Kruskal’s theorem on tensor decomposition},
journal = {Linear Algebra and its Applications},
volume = {432},
number = {7},
pages = {1818-1824},
year = {2010},
issn = {0024-3795},
doi = {https://doi.org/10.1016/j.laa.2009.11.033},
url = {https://www.sciencedirect.com/science/article/pii/S0024379509006132},
author = {John A. Rhodes}
}

@article{ALS,
author = {Uschmajew, Andr\'{e}},
title = {Local Convergence of the Alternating Least Squares Algorithm for Canonical Tensor Approximation},
journal = {SIAM Journal on Matrix Analysis and Applications},
volume = {33},
number = {2},
pages = {639-652},
year = {2012},
doi = {10.1137/110843587},

URL = {https://doi.org/10.1137/110843587}
}

@article{gradient_descent,
author = {Kolda, Tamara G. and Hong, David},
title = {Stochastic Gradients for Large-Scale Tensor Decomposition},
journal = {SIAM Journal on Mathematics of Data Science},
volume = {2},
number = {4},
pages = {1066-1095},
year = {2020},
doi = {10.1137/19M1266265},

URL = {https://doi.org/10.1137/19M1266265}
}

@article{Canay_testability2013,
author = {Canay, Ivan A. and Santos, Andres and Shaikh, Azeem M.},
title = {On the Testability of Identification in Some Nonparametric Models With Endogeneity},
journal = {Econometrica},
volume = {81},
number = {6},
pages = {2535-2559},
doi = {https://doi.org/10.3982/ECTA10851},
url = {https://onlinelibrary.wiley.com/doi/abs/10.3982/ECTA10851},
year = {2013}
}

@article{Hu_Completeness,
author = {Yingyao Hu and Ji-Liang Shiu},
title = {A simple test of completeness in a class of nonparametric specification},
journal = {Econometric Reviews},
volume = {41},
number = {4},
pages = {373--399},
year = {2022},
publisher = {Taylor \& Francis},
doi = {10.1080/07474938.2021.1957285},
URL = { https://doi.org/10.1080/07474938.2021.1957285
}
}

@book{Brown1986ExponentialFamilies,
  author    = {Brown, Lawrence D.},
  title     = {Fundamentals of Statistical Exponential Families},
  year      = {1986},
  publisher = {Institute of Mathematical Statistics},
  series    = {Lecture Notes--Monograph Series},
  address   = {Hayward, CA}
}

@incollection{elwert2013graphical,
  author    = {Elwert, Felix},
  title     = {Graphical Causal Models},
  booktitle = {Handbook of Causal Analysis for Social Research},
  editor    = {Morgan, Stephen L.},
  publisher = {Springer Netherlands},
  address   = {Dordrecht},
  year      = {2013},
  pages     = {245--273}
}

@article{tchetgen2024proximal,
  author = {Tchetgen Tchetgen, Eric J. and Ying, Andrew and Cui, Yifan and Shi, Xu and Miao, Wang},
  title = {An Introduction to Proximal Causal Inference},
  journal = {Statistical Science},
  volume = {39},
  number = {3},
  pages = {375--390},
  year = {2024},
  month = {August},
  publisher = {Institute of Mathematical Statistics}
}

@article{Derksen2013,
  author  = {Harm Derksen},
  title   = {Kruskal’s uniqueness inequality is sharp},
  journal = {Linear Algebra and its Applications},
  volume  = {438},
  number  = {2},
  pages   = {708--712},
  year    = {2013},
  note    = {Tensors and Multilinear Algebra},
  issn    = {0024-3795},
  doi     = {10.1016/j.laa.2011.05.041},
  url     = {https://www.sciencedirect.com/science/article/pii/S0024379511004873}
}

\end{document}